\documentclass[12pt]{article}

\usepackage[utf8]{inputenc}

\usepackage{amsmath,amsthm,amssymb}
\usepackage{mathrsfs}
\usepackage{authblk}

\theoremstyle{plain}
\newtheorem{theorem}{Theorem}
\newtheorem{lemma}{Lemma}
\newtheorem{proposition}{Proposition}
\theoremstyle{definition}
\newtheorem{corollary}{Corollary}
\theoremstyle{remark}

\newtheorem*{remark*}{Remark}

\numberwithin{equation}{section}

\usepackage[margin=25mm]{geometry}

\usepackage{fancyhdr}
\begin{document}

\title{Stable subspaces of positive maps of matrix algebras}

\author[1]{Marek Miller\thanks{marek.miller@ift.uni.wroc.pl}}
\author[1]{Robert Olkiewicz\thanks{robert.olkiewicz@ift.uni.wroc.pl}}
\affil[1]{Institute of Theoretical Physics, Uniwersytet Wroc{\l}awski, Poland}

\date{}

\maketitle

\abstract{ 
We study stable subspaces of positive extremal maps of finite dimensional
matrix algebras that preserve trace and matrix identity 
(so-called bistochastic maps).
We have established the existence of the isometric-sweeping
decomposition for such maps.
As the main result of the paper,
we have shown that all extremal bistochastic maps
acting on the algebra of matrices of size 3x3
fall into one of the three possible categories,
depending on the form of the stable subspace of
the isometric-sweeping decomposition.
Our example of an extremal atomic positive map
seems to be the first one that handles the case
of that subspace being non-trivial.
Lastly, we compute the entanglement witness associated with the extremal map
and specify a large family of entangled states detected by it.

\vspace{0.1cm}

}

\section*{Introduction}
\label{sec:Introduction}

\paragraph{}
Positive maps of operator algebras is an increasingly
popular subject of research,
both as a mathematical theory interesting in itself,
as well as a prominent domain of applications to quantum theory
\cite{stormer2013positive}.
Started in the pioneering work of Kadison
\cite{kadison1952generalized},
and Stinespring
\cite{stinespring1955positive},
the theory reached its first major breakthrough,
when St{\o}rmer and Woronowicz presented a structure theorem
for positive maps of low dimensional matrix algebras
\cite{stormer1963positive,woronowicz1976positive}.
It is also necessary to mention the work of \mbox{M.-D.\,Choi} here,
as the major contribution to the theory at every stage
\mbox{\cite{choi1975completely,choi1980some,choi1977extremal}}.
The second turning point came in 1990s,
when Peres and P.\,M.\,R.\,Horodeckis pointed out at the intrinsic relation
between separable states of composite quantum systems
and positive maps of algebras of observables
\cite{peres1996separability,horodecki1996separability}.
It turns out that there is a one-to-one correspondence between
positive maps and entanglement witnesses
\cite{chruscinski2014entanglement},
and the Peres-Horodecki criterion,
whether a quantum state is separable,
is computationally feasible as long as a structure theorem
similar to that proven by St{\o}rmer and Woronowicz holds true.
Unfortunately,
whereas that correspondence does exists even for
the most general infinite-dimensional quantum systems
\mbox{\cite{stormer2009separable,miller2014horodeckis}},
higher dimensional situation lacks
the complete description of positive maps
and one needs a deeper understanding of their highly nontrivial structure.
Because the set of positive maps forms a convex cone,
its elements can be characterised as convex combinations of extremal ones.
This is a consequence of the celebrated Krein-Milman theorem
\cite{krein1940extreme}, as for every positive number $r > 0$,
the set of positive maps such that their operator norm: $||S|| \leq r$,
is compact.
This is certainly true for maps on finite-dimensional matrix algebras,
but applies also to the general setting of von Neumann algebras \cite{miller2014horodeckis}.
Despite considerable effort,
examples of extremal positive maps,
even in the low-dimensional case,
are scarce
\cite{choi1977extremal,osaka1992class,ha2011entanglement,chruscinski2011exposed}.
In this paper, we deal exclusively with extremal maps as extreme points in
the cone of all positive maps on matrix algebras (see below).

In order to facilitate the further study of extremal positive maps,
we propose to take a closer look at their stable subspaces.
In general, by a 'stable subspace' of a positive map 
$S \! : M_{n} \rightarrow M_{n}$,
where $M_{n}$ stands for the algebra of square complex matrices
of size $n$,
we mean a particular subspace $K_{S} \subset M_{n}$
(in fact, a JB$^{*}$-algebra, see 
\mbox{Corollary \ref{cor:KisJordanAlgebra}} 
below),
such that the map is a Jordan automorphism on $K_{S}$ and 
$S^{k} \rightarrow 0$ strongly
on the orthogonal complement of $K_{S}$,
as $k \rightarrow \infty$
(see the more precise definition below).
We say that the map $S$ is strongly ergodic,
if $K_{S}$ is isomorphic to the set of complex numbers.
Adapting the previous results in this matter from
\cite{olkiewicz1999environment},
we establish the existence of the isometric-sweeping
decomposition for a positive map that preserves trace and
the matrix identity
(so-called bistochastic maps).
Such decomposition says precisely that $S$ must be
a Jordan automorphism of $K_{S}$,
and its powers tend strongly to $0$ for any matrix orthogonal to $K_{S}$.
What is interesting in this respect,
is that if we narrow our interest to the maps acting on $M_{3}$,
it turns out that there are only three possible
forms the algebra $K_{S}$ can have.
Moreover, it seems that all previously known examples of extremal positive
maps of $M_{3}$ fall into only two of the three possible categories:
in particular,
it is easy to see that for a Jordan automorphism $S(A) = U^{*} A \, U$,
or $S(A) = U^{*} A^{t} \, U$,
where $U$ is unitary,
$K_{S} = M_{3}$;
whereas the Choi map
\cite{osaka1992class,choi1977extremal}
is strongly ergodic.
The remaining possibility among those three,
namely the two dimensional commutative subalgebra,
requires from us to expand on an example of a positive map
that cannot be classified as belonging to one of the types
previously encountered in the literature.
One of the main results of this paper is the proof
that the map provided by us as the example is indeed extremal and atomic
\cite{ha1998atomic}.
Moreover, we show that the map is also exposed.
This fact gives immediately a useful entanglement witness
that the authors hope could serve as a prototypic case in further study 
of the entanglement of three-level quantum systems
\cite{bertlmann2005optimal,caves2000qutrit}.

The structure of the paper is as follows.
In Section \ref{sec:Preliminaries},
we introduce some basic notions and notation.
In Section \ref{sec:Decomposition},
following \cite{olkiewicz1999environment},
we establish the existence of the isometric-sweeping
decomposition of bistochastic maps of finite dimensional
matrix algebra $M_{n}$.
Next, in Section \ref{sec:ExtremalMaps},
we provide the main result of this paper and prove that
for bistochastic map of $M_{3}$
the stable subspace of the decomposition takes only one of three
possible forms.
We end the paper with an original example of an extremal positive map of 
the algebra $M_{3}$;
we compute the entanglement witness associated with the extremal map
and specify a large family of entangled states detected by it.

\section{Preliminaries}
\label{sec:Preliminaries}

\paragraph{}
Let $n \in \mathbb{N}$.
The complex linear space $\mathbb{C}^{n}$ consists of column vectors:
$\eta, \xi$, etc.; the respective row vectors are denoted
$\eta^{t}, \xi^{t}$, etc.
The bar over a number, vector or a matrix always denotes
the element-wise complex conjugation: e.g. 
$\overline{z}$, $\overline{\eta}$, $\overline{A}$, etc.
The space $\mathbb{C}^{n}$ is equipped with the inner product
$\langle \eta, \xi \rangle = \eta^{*} \xi$,
where $\eta^{*} = \overline{\eta}^{\,t}$,
and $||\eta||^{2} = \eta^{*} \eta$.
To avoid confusion when two different spaces,
say $\mathbb{C}^{m}$ and $\mathbb{C}^{n}$, are involved,
we distinguish vectors belonging to one of them with an arrow:
e.g. $\vec{\eta}$, $\vec{\xi}$.
Let $M_{n} = M_{n}(\mathbb{C})$ be the algebra
of square complex matrices of size $n$.
We will use letters $A, B, C$, etc. to specify a matrix of $M_{n}$.
The norm of $A$, denoted by $||A||$, is understood as the operator 
norm of $A$ as a linear map acting on $\mathbb{C}^{n}$.
For a vector $\eta \in \mathbb{C}^{n}$, $\eta \neq 0$,
by $P_{\eta}$ or $P({\eta}) \in M_{n}$,
we denote the rank-one operator
$P_{\eta}  =  \eta \eta^{*}$.
Of course, when $||\eta|| = 1$, 
$P_{\eta}$ is the orthogonal projection onto one-dimensional space
spanned by $\eta$.
For $A \in M_{n}$, 
we denote its trace by $\text{Tr} A$;
and by $A^{t}$ and $A^{*} = \overline{A}^{\,t}$ 
its transpose and conjugate transpose, respectively.
We say that a matrix $A$ is positive-semidefinite,
or simply \emph{positive},
if $\eta^{*} A \eta \geq 0$ for any $\eta \in \mathbb{C}^{n}$
(i.e. $A = A^{*}$ and $A$ has a non-negative spectrum).

A linear map $S\!: M_{n} \rightarrow M_{n}$ is said to be positive,
indicated as $S \geq 0$,
if for any $A \in M_{n}$ such that $A \geq 0$,
we have $S(A) \geq 0$. 
For a positive map $S$,
its operator norm is given by $||S|| = S(\mathbf{1})$,
where $\mathbf{1}_{n}$, or simply $\mathbf{1}$,
is the identity matrix of $M_{n}$.
Any positive map is Hermitian,
i.e.  $S(A^{*}) = S(A)^{*}$, for all $A \in M_{n}$.
The identity map of $M_{n}$ is labelled $I_{n}$,
or simply $I$.
The convex cone of all positive maps of $M_{n}$ is denoted
by $\mathcal{P}(M_{n})$.
For $k \in \mathbb{N}$, a map $S \in \mathcal{P}(M_{n})$,
such that the map
$I_{k} \otimes S :  M_{k} \! \otimes \! M_{n}
 \rightarrow  M_{k} \! \otimes \! M_{n}$
is positive,
is called $k$-positive.
If a map is a $k$-positive map for every $k$,
it is called completely positive.
Similarly, a map is $k$-\emph{co}positive,
or completely copositive,
if $I_{k} \otimes (S \circ t)$ is positive for
some $k$, or for every $k$, respectively,
where $t \! : A \mapsto A^{t}$,
$A \in M_{n}$, is the transposition map.
A positive map $S$ is called decomposable,
if it can be written in the form
$S(A) = \Lambda_{1}(A) + \Lambda_{2} (A^{t})$,
where both maps $\Lambda_{1}, \Lambda_{2}$ are completely positive,
possibly zero
(see \cite{choi1975completely} for more details on completely positive maps).
It is called atomic \cite{ha1998atomic},
if it cannot be written as a sum of 2-positive and 2-copositive maps.
A positive map $S$ is extremal,
if for any positive map $T: M_{n} \rightarrow M_{n}$ such that
$S - T \in \mathcal{P}(M_{n})$,
i.e. $0 \leq T \leq S$,
we have $T = \alpha S$ for some number $0 \leq \alpha \leq 1$.
We denote the set of extremal maps of $\mathcal{P}(M_{n})$
by $\text{Ext}(M_{n})$.
It is true that every positive map
can be written as a convex combination of extremal ones.
Nevertheless, it could be of use to specify a dense subset
of $\text{Ext}(M_{n})$,
such that it is easier to handle its elements instead of
general extremal maps.
Hence, we say that a positive map $S$ of the matrix algebra
$M_{n}$ is exposed
\cite{marciniak2013rank},
if for all $T \in \mathcal{P}(M_{n})$ such that
$\text{Tr}\, P_{\xi} \, T(P_{\eta}) = 0$,
where $\mathbb{C}^{n} \ni \eta, \xi \neq 0$ is any pair for which
$\text{Tr}\, P_{\xi} \, S(P_{\eta}) = 0$,
we have that $T = \alpha S$, $\alpha \geq 0$.
Due to the Straszewicz theorem \cite{straszewicz1935exponierte},
the set of exposed maps of $M_{n}$ is indeed dense in $\text{Ext}(M_{n})$.

By a JB$^{*}$-algebra $K$ we understand a complex Banach space 
which is also a complex Jordan algebra.
We always assume that $\mathbf{1} \in K$.
See 
\cite{hanche1984jordan}
for the overview of the theory of JB$^{*}$-algebras.

\section{Decomposition of bistochastic maps}
\label{sec:Decomposition}

\paragraph{}
Let $S\!: M_{n} \rightarrow M_{n}$ be a positive map such that
$S(\mathbf{1}) = \mathbf{1}$,
$\text{Tr} S(A) = \text{Tr} A$,
for any $A \in M_{n}$.
We call such a map \emph{bistochastic}.
It is easy to see that $S$ is a contraction in
the Hilbert-Schmidt norm (\emph{HS-norm}) on $M_{n}$,
defined as
$||A||_{HS} = \left( \text{Tr} \, A^{*} A \right)^{1/2}$.
Indeed,
since $S$ fulfils the Kadison-Schwarz inequality,
\begin{equation}
\label{eq:SchwarzInequality}
    S(A^{*}) \, S(A) \: \leq \: S(A^{*} A),
\end{equation}
for any normal element $A \in M_{n}$ (see Prop.\,3.6 of \cite{choi1980some}),
assuming at first that $A = A^{*}$,
we have
\begin{equation}
\label{RandomLabel:712874}
    || S(A) ||_{HS}^{2} \: = \: \text{Tr} S(A)^{2}
    \: \leq \:  \text{Tr} S(A^{2}) \: = \:
        \text{Tr} A^{2} \: = \: ||A||_{HS}^{2}.
\end{equation}
By representing a general element $A \in M_{n}$ as a sum
$A = A_{1} + i A_{2}$,
where both $A_{1}, A_{2}$ are self-adjoint,
and repeating essentially the same calculation as in
\eqref{RandomLabel:712874},
we obtain the assertion.

Next,
we are going to define an isometric splitting of $S$,
drawing from the ideas presented in \cite{olkiewicz1999environment}.
Because $S$ is a contraction in the HS-norm on the Hilbert space $M_{n}$,
equipped with the Hilbert-Schmidt inner product,
that space can be decomposed into a direct sum of a space $K_{S}$,
defined as
\begin{equation}
\label{def:definitionofK}
    K_{S} \: = \: \left\{
        A \in M_{n}: \,\,
            || S^{k} A ||_{HS} = || S^{* k} A ||_{HS} = || A ||_{HS}, \:\:
            \forall k \in \mathbb{N}
        \right\},
\end{equation}
and its orthogonal complement $K_{S}^{\perp}$.
By $S^{*}$, we denote the bistochastic map of $M_{n}$,
which is the adjoint of $S$ as a linear operator on the Hilbert space $M_{n}$,
i.e. $\text{Tr}\, S^{*} (A)\, B = \text{Tr} A \, S(B)$,
for every $A, B \in M_{n}$.
It is true that $A \in K_{S}$,
if and only if
$S^{* k} S^{k} A = S^{k} S^{* k} A = A$,
for any $k \in \mathbb{N}$.
The following proposition,
taken with slight modification from
\cite{olkiewicz1999environment},
puts together some of the characteristics of $K_{S}$.

\begin{proposition}
\label{prop:propertiesofK}
Suppose $S$ is a bistochastic map of $M_{n}$,
and the subspace $K_{S}$ is defined as in \eqref{def:definitionofK}.
Then:

\let \oldlabelenumi \labelenumi
\renewcommand{\labelenumi}{\alph{enumi})}
\begin{enumerate}
\item $\mathbf{1} \in K_{S}$;
\item $A \in K_{S}$ implies that $A^{*} \in K_{S}$;
\item  $A  \in K_{S}$ implies that
$|A| = (A^{*}A)^{1/2} \in K_{S}$;
\item if $A, B \in K_{S}$, then $AB + BA \in K_{S}$;
\item if $A = A^{*} \in K_{S}$ and $A = \sum_{i=1}^{k} \lambda_{i} P_{i}$,
where each $\lambda_{i} \neq 0$ is different,
and each $P_{i}$ is a orthogonal projection in $M_{n}$,
then $P_{i} \in K_{S}$ for any $i = 1,2,\ldots,k$;
\item if $P \in K_{S}$ is an orthogonal projection,
then $S(P)$ and $S^{*}(P)$ are orthogonal projections as well,
and $\mathrm{dim} \, S(P) = \mathrm{dim} \, S^{*}(P) = \mathrm{dim} \, P$;
\item if $P, Q \in K_{S}$ are orthogonal projections such that $P Q = 0$,
then $S(P) S(Q) = S^{*}(P) S^{*}(Q) = 0$.
\end{enumerate}
\let \labelenumi \oldlabelenumi
\end{proposition}

\begin{proof}
a) obvious, since $S$ is bistochastic;
b), c) as in Prop. 5\,a) and 5\,b) of \cite{olkiewicz1999environment}.

d)  Suppose at first that $A = A^{*} \in K_{S}$.
Now, $S^{*k} S^{k}(A) = S^{k} S^{*k}(A) = A$,
and, as a result of the Kadison-Schwarz inequality applied to the map
$S^{*k} S^{k}$,
\begin{equation}
A^{2} \: = \: \left( S^{*k} S^{k}(A) \right) \left( S^{*k} S^{k}(A) \right)
    \: \leq \: S^{*k} S^{k}(A^{2}).
\end{equation}
Hence
\begin{equation}
||A^{2}||_{HS} \: \leq \: || S^{*k} S^{k}(A^{2}) ||_{HS} \: \leq \:
    || S^{k}(A^{2}) ||_{HS} \: \leq \: ||A^{2}||_{HS},
\end{equation}
i.e. $|| S^{k}(A^{2}) ||_{HS} = ||A^{2}||_{HS}$.
By a similar argument
$|| S^{*k}(A^{2}) ||_{HS} = ||A^{2}||_{HS}$,
and thus $A^{2} \in K_{S}$.
Now, for any $A, B \in K_{S}$, such that $A = A^{*}$, $B = B^{*}$,
we have
$AB + BA = (A + B)^{2} - A^{2} - B^{2} \in K_{S}$.
For a general $A \in K_{S}$, we write
$A = A_{1} + i A_{2}$, where $A_{1}, A_{2}$ are both Hermitian.
Then
$A^{2} = A_{1}^{2} - A_{2}^{2} + i(A_{1} A_{2} + A_{2} A_{1}) \in K_{S}$.
It is now evident that $AB+BA \in K_{S}$ for general $A,B \in K_{S}$.

e) as in Prop. 5\,d);
f), g) as in Prop. 6\,a) and 6\,c) of
\cite{olkiewicz1999environment}.
\end{proof}

From what has been said above, 
follows immediately the next observation.

\begin{corollary}
\label{cor:KisJordanAlgebra}
The space $K_{S}$, with the matrix norm $|| \cdot ||$,
and the multiplication $A \circ B = \frac{1}{2}(AB + BA)$,
is a JB$^{*}$-algebra.
The map $S$ is a Jordan automorphism on $K_{S}$ and 
\begin{equation}
\label{eq:SGoesTo0OnKOrth}
    \lim \limits_{k\rightarrow \infty} S^{k}(A)  \: = \:
    \lim \limits_{k\rightarrow \infty} S^{*k}(A) \: = \: 0,
\end{equation}
for any $A \in K_{S}^{\perp}$.
\end{corollary}

\begin{proof}
It is obvious that $K_{S}$ is a complex Banach space. 
From Proposition \ref{prop:propertiesofK}\,b),
we have that $K_{S}$ is equipped with involution,
and from point d) 
that $K_{S}$ is a Jordan algebra.
Since the map $S$ is invariant with respect to the subspaces
$K_{S}$ and $K_{S}^{\perp}$,
and of course $M_{n} = K_{S} \oplus K_{S}^{\perp}$,
we see that $S$ splits into direct sum
$S = S_{1} \oplus S_{2}$,
where $S_{1} = S_{| K_{S}}$, $S_{2} = S_{| K_{S}^{\perp}}$.
From Prop. 7a) of
\cite{olkiewicz1999environment},
we have that $S(A^{*} A) = S(A)^{*} S(A)$,
i.e. $S$ is a Jordan homomorphism on $K_{S}$.
Moreover,
because $S^{*} S = S S^{*} = I$ on $K_{S}$,
the map $S_{1}$ is invertible,
and thus it is a Jordan automorphism
(see Definition 3.2.1(6) in \cite{bratteli2003operator}).
From the definition \eqref{def:definitionofK} of the space $K_{S}$,
follows easily that 
eq. \eqref{eq:SGoesTo0OnKOrth} holds.
\end{proof}

\begin{theorem}
\label{thm:FromESbook}
Let $K \subset M_{n}$ be a JB$^{*}$-subalgebra of $M_{n}$.
There is a bistochastic map
$S \! : M_{n} \rightarrow M_{n}$ such that $K = K_{S}$.
\end{theorem}
\begin{proof}
We consider the space $M_{n}$ to be a Hilbert space equipped
with the Hilbert-Schmidt inner product.
Let $S\!: M_{n} \rightarrow K \subset M_{n}$ be the 
orthogonal projection onto $K$.
It is evident that $S(\mathbf{1}) = \mathbf{1}$,
and because $S = S^{*}$, 
the map $S$ also preserves trace and
the space $K$ is a stable subspace for $S$,
provided $S$ is a positive map.
Let then $A \in M_{n}$ be a positive matrix.
Since $K$ is JB$^{*}$-algebra,
$(SA)^{*} \in K$.
Because 
$|| A - SA ||_{HS} = || A - (SA)^{*} ||_{HS}$,
and $SA$ is the best approximation of $A$ in the space $K$,
we have that $(SA)^{*} = SA$.
We write then $SA = B_{+} - B_{-}$,
where both $B_{+}, B_{-}$ are positive and 
$B_{+} B_{-} = 0$.
We have assumed that $\mathbf{1} \in K$,
and because $(SA)^{k} \in K$ for any $k \in \mathbb{N}$,
we have that also the modulus $|S A| \in K$
(compare the explicit formula for
the square-root of a positive matrix in \cite{bratteli2003operator}, p. 34).
Hence both $B_{+}, B_{-} \in K$.
We compute
\begin{multline}
|| A - SA ||_{HS}^{2} \: = \: \text{Tr} \, ( A - SA )^{2} \: = \:
    \text{Tr} \, ( A - B_{+} + B_{-} )^{2} \: = \: \\
    \text{Tr} \, ( A - B_{+} )^{2} + 
        2 \, \text{Tr} \, ( A - B_{+} ) B_{-} + \text{Tr} \, B_{-}^{2}
            \: = \: \\
    || A - B_{+} ||_{HS}^{2} + 
        2 \, \text{Tr} \,  A \, B_{-}+ \text{Tr} \, B_{-}^{2} 
    \: \geq \: || A - B_{+} ||_{HS}^{2}.
\end{multline}
Again, since $SA$ is the best approximation of $A$ in $K$,
we have $SA = B_{+}$, i.e. $SA \geq 0$, which ends the proof.
\end{proof}

From the above,
we know that $K_{S}$ has additional structure of a JB$^{*}$-algebra.
We will see in the following that,
with additional assumption imposed on S,
this algebra must necessarily
have a specific structure,
at least for low dimensional matrices.

\section{Extremal positive maps}
\label{sec:ExtremalMaps}

\paragraph{}
In this section, we focus on maps of the algebra $M_{3}$.
For convenience, let us label four orthogonal projections:
\begin{equation}
\label{def:OrthogonalProjections}
    P_{1} = \begin{pmatrix}
        1 & 0 & 0 \\
        0 & 0 & 0 \\
        0 & 0 & 0
    \end{pmatrix} , \quad \quad
    P_{2} = \begin{pmatrix}
        0 & 0 & 0 \\
        0 & 1 & 0 \\
        0 & 0 & 0
    \end{pmatrix} , \quad \quad
    P_{3} = \begin{pmatrix}
        0 & 0 & 0 \\
        0 & 0 & 0 \\
        0 & 0 & 1
    \end{pmatrix},
\end{equation}
and $P_{12} = P_{1} + P_{2}$.

It is of importance to us that
Theorems 5.3.8 and 6.2.3 of \cite{hanche1984jordan}
imply that every JB$^{*}$-algebra contained in $M_{3}$ is
isomorphic to one of the following:
$\mathbb{C}\mathbf{1}$, 
$\mathbb{C} P_{12} \oplus \mathbb{C} P_{3}$,
$\mathbb{C} P_{1} \oplus \mathbb{C} P_{2} \oplus \mathbb{C} P_{3}$,
$M_{2} \oplus \mathbb{C} P_{3}$,
$M_{2}^{s} \oplus P_{3}$, 
$M_{3}^{s}$,
and
$M_{3}$ itself,
where $M_{n}^{s}$ is the Jordan algebra of symmetric matrices of 
size $n$: $M_{2}^{s} = \{ A \in M_{n}: A = A^{t} \}$.  

\begin{theorem}
\label{thm:ExposedMaps}
Let $S\!: M_{3} \rightarrow M_{3}$ be an extremal bistochastic map.
Then the JB$^{*}$-algebra $K_{S}$ is isomorphic to one of the following:
$\mathbb{C}\mathbf{1}$, $\mathbb{C} P_{12} \oplus \mathbb{C} P_{3}$,
or
$M_{3}$.
\end{theorem}

\begin{proof}
1. We proof the assertion by excluding all other possible forms of
$K_{S}$ in the first place.
Let $K_{S}$ be one of the following JB$^{*}$-algebras:
$\mathbb{C} P_{1} \oplus \mathbb{C} P_{2} \oplus \mathbb{C} P_{3}$,
$M_{2} \oplus \mathbb{C} P_{3}$,
$M_{2}^{s} \oplus P_{3}$, or $M_{3}^{s}$.
Then $K_{S}$ contains the projections $P_{1}, P_{2}, P_{3}$,
and because $S$ is a Jordan automorphism on $K_{S}$:
\begin{multline}
    \text{Tr} \, S(P_{i}) S(P_{j}) = \\
= \frac{1}{2} \text{Tr} \, \left( S(P_{i}) S(P_{j}) + S(P_{j}) S(P_{i}) \right)=
\frac{1}{2} \text{Tr} \, S(P_{i} P_{j} + P_{j} P_{i})=
\text{Tr} \, P_{i} P_{j} = \delta_{ij},
\end{multline}
where $i,j = 1,2,3$ and $\delta_{ij}$ is the Kronecker delta.  
Hence, $\left\{S(P_{i})\right\}_{i=1}^{3}$ 
is a triple of rank-one, mutually orthogonal projections.
There is a unitary matrix $U \in M_{3}$, such that
$U^{*} S(P_{i}) U = P_{i}$, for $i = 1,2,3$. 
Let us define $\tilde{S}(A) = U^{*} S(A) U$.
Then $\tilde{S}$ is an extremal 
\cite[Lemma 3.1.2b, p.\,27]{stormer2013positive}
bistochastic map such that
$\tilde{S}(P_{i}) = P_{i}$, $i = 1,2,3$.
By \cite[Theorem 4.1]{kye1995positive},
$\tilde{S}$ is decomposable,
which contradicts either the fact that $S$ is extremal
or $K_{S} \neq M_{3}$.

2.
While it is easy to see that there are extremal maps for which
$K_{S}=M_{3}$ (take e.g. $S(A) = U A U^{*}$ for a unitary matrix $U$),
or $K_{S}=\mathbb{C} \mathbf{1}$
(take the celebrated Choi map
\cite{choi1977extremal});
it is not that straightforward to provide an example of
an extremal bistochastic map which has
$K_{S}= \mathbb{C} P_{12} \oplus \mathbb{C} P_{3}$.

Let then $S\!: M_{3} \rightarrow M_{3}$ be a linear map defined as
\begin{equation}
\label{eq:DefinitionOfS}
S(A) \:=\: \begin{pmatrix}
        \frac{1}{2}(a_{11} + a_{22}) & 0 & \frac{1}{\sqrt{2}} a_{13} \\
        0 & \frac{1}{2}(a_{11} + a_{22}) & \frac{1}{\sqrt{2}} a_{32} \\
        \frac{1}{\sqrt{2}} a_{31} & \frac{1}{\sqrt{2}} a_{23} & a_{33}
        \end{pmatrix},
\end{equation}
for
$A = \left( a_{ij} \right)_{i,j=1}^{3}
        \in M_{3}$.
If we use the notation
$A = \left( \begin{smallmatrix}
    B & \vec{u} \\
    \vec{w}^{t} & z
    \end{smallmatrix} \right)$,
for $B \in M_{2}$,
$\vec{u},\vec{w} \in \mathbb{C}^{2}$ are column vectors,
and $z \in \mathbb{C}$,
the map $S$ acts by
\begin{equation}
    S(A) \:=\: S \begin{pmatrix}
    B & \vec{u} \\
    \vec{w}^{t} & z
    \end{pmatrix} \: = \:
    \begin{pmatrix}
        \frac{1}{2} (\text{Tr} B) \, \mathbf{1}_{2} &
            \frac{1}{\sqrt{2}}(\hat{P}_{1} \vec{u} + \hat{P}_{2} \vec{w}) \\
        \frac{1}{\sqrt{2}}(\hat{P}_{1} \vec{w} + \hat{P}_{2} \vec{u})^{t} & z
    \end{pmatrix},
\end{equation}
where
$\hat{P}_{1} = \left( \begin{smallmatrix} 1 & 0 \\ 0 & 0 \end{smallmatrix} \right)$,
$\hat{P}_{2} = \left( \begin{smallmatrix} 0 & 0 \\ 0 & 1 \end{smallmatrix} \right)$,
and $\mathbf{1}_{2}$ is the identity matrix of $M_{2}$.
Provided the next Lemma \ref{lem:SIsExtremal} is true, 
this example ends the proof.
\end{proof}

\begin{lemma}
\label{lem:SIsExtremal}
$S$ is an bistochastic, extremal and atomic map.
\end{lemma}

\begin{proof}
In order to prove that $S$ is a positive map,
it is enough to show that for any $\eta \in \mathbb{C}^{3}\backslash\{0\}$, 
$SP_{\eta} \geq 0$,
where $P_{\eta}$ is a rank-one operator,
$P_{\eta} = \eta \eta^{*}$.
Let us then take
$\eta = (\eta_{1}, \eta_{2}, \eta_{3}) = (\vec{\eta}, \eta_{3})$,
$\vec{\eta} = (\eta_{1}, \eta_{2}) \in \mathbb{C}^{2}\backslash\{0\}$,
$\eta_{3} \in \mathbb{C}$.
We have then
\begin{multline}
 SP_{\eta} \:=\:  S \begin{pmatrix}
    \vec{\eta} \vec{\eta}^{\,*} & \overline{\eta}_{3} \vec{\eta} \\
    \eta_{3} \vec{\eta}^{\,*}   & |\eta_{3}|^{2}
 \end{pmatrix} \: = \: \\
 = \begin{pmatrix}
  \frac{||\vec{\eta}||^{2}}{2} \mathbf{1}_{2} &
        \frac{1}{\sqrt{2}} \left ( \overline{\eta}_{3} \hat{P}_{1} \vec{\eta} +
          \eta_{3} \hat{P}_{2} \overline{\vec{\eta}} \right) \\
\frac{1}{\sqrt{2}} \left ( \overline{\eta}_{3} \hat{P}_{1} \vec{\eta} +
          \eta_{3} \hat{P}_{2} \overline{\vec{\eta}} \right)^{*} &
        |\eta_{3}|^{2}
 \end{pmatrix}.
\end{multline}
If $\eta_{3} = 0$, then of course $SP_{\eta} \geq 0$.
In the case when $\eta_{3} \neq 0$,
taking the Schur complement
(see \cite[Theorem 1.12, p.34]{zhang2006schur}),
we have that $SP_{\eta} \geq 0$, if and only if
\begin{equation}
\label{ieq:SchurForS}
   \left ( \overline{\eta}_{3} \hat{P}_{1} \vec{\eta} +
    \eta_{3} \hat{P}_{2} \overline{\vec{\eta}} \right)
   \left ( \overline{\eta}_{3} \hat{P}_{1} \vec{\eta} +
    \eta_{3} \hat{P}_{2} \overline{\vec{\eta}} \right)^{*} 
    \: \leq \:
        |\eta_{3}|^{2} \, ||\vec{\eta}||^{2} \, \mathbf{1}_{2},
\end{equation}
but it is easy to see that this inequality is fulfilled:
\begin{multline}
   \left ( \overline{\eta}_{3} \hat{P}_{1} \vec{\eta} +
    \eta_{3} \hat{P}_{2} \overline{\vec{\eta}} \right)
   \left ( \overline{\eta}_{3} \hat{P}_{1} \vec{\eta} +
    \eta_{3} \hat{P}_{2} \overline{\vec{\eta}} \right)^{*} \: \leq \:
|| \overline{\eta}_{3} \hat{P}_{1} \vec{\eta} +
    \eta_{3} \hat{P}_{2}\overline{\vec{\eta}} ||^{2} \, \mathbf{1}_{2}
    \: = \: \\
\left( || \overline{\eta}_{3} \hat{P}_{1} \vec{\eta} ||^{2} +
    || \eta_{3} \hat{P}_{2} \overline{\vec{\eta}} ||^{2} \right) \, \mathbf{1}_{2} 
    \: = \:
|\eta_{3}|^{2} \, ||\vec{\eta}||^{2} \, \mathbf{1}_{2}.
\end{multline}
Thus, $S$ is a bistochastic map.
We also have
$K_{S}= \mathbb{C} P_{12} \oplus \mathbb{C} P_{3}$.
Indeed,
one computes immediately that for $A \in M_{3}$:
\begin{equation}
\lim \limits_{k \rightarrow \infty} S^{k}(A) \:=\:
\frac{1}{2} (\text{Tr} \, P_{12} A )\, P_{12} +
    (\text{Tr} \, P_{3} A)\, P_{3} \, \in K_{S}.
\end{equation}
From that and from the definition \eqref{def:definitionofK},
follows the particular form of the JB$^{*}$-algebra $K_{S}$.

We are going to show now that $S$ is extremal.
Let $S_{0}: M_{3} \rightarrow M_{3}$ be a positive map such that
$0 \leq S_{0} \leq S$.
We have $S_{0}(P_{3}) \leq S(P_{3}) = P_{3}$,
and hence $S_{0}(P_{3}) = \alpha P_{3}$ for some $0 \leq \alpha \leq 1$.
Let us consider the space $M_{2}$ embedded into $M_{3}$ as follows:
$A \in M_{2} \! \subset \! M_{3}$ whenever
\begin{equation}
\label{RandomLabel:450031}
    A \: = \: \begin{pmatrix}
    a_{11} & a_{12} & 0 \\
    a_{21} & a_{22} & 0 \\
    0 & 0 & 0
    \end{pmatrix}, \quad
    a_{ij} \in \mathbb{C}, \,\, i,j = 1,2.
\end{equation}
We want to show that $S_{0}(A) \in M_{2} \! \subset \! M_{3}$
for any $A \in M_{2} \! \subset \! M_{3}$.
Let $B \in M_{2}$ and suppose at first that $B \geq 0$.
Then
\begin{equation}
\label{eq:SMapsM2intoM2}
 0 \:\leq\: S_{0} \begin{pmatrix}
               B & \vec{0} \\ \vec{0}^{t} & 0
              \end{pmatrix} \: = \:
 \begin{pmatrix}
  \hat{S}_{0}(B) & \vec{u} \\ \vec{w}^{t} & r
 \end{pmatrix},
\end{equation}
for some vectors $\vec{u}, \vec{w} \in \mathbb{C}^{2}$ and $r \geq 0$.
The map $\hat{S}_{0}:B \mapsto \hat{S}_{0}(B) \in M_{2}$ must be a positive
map of $M_{2}$ such that $\hat{S}_{0}(B) \leq \frac{1}{2} (\text{Tr} B) \mathbf{1}_{2}$.
On the other hand,
\begin{equation}
 0 \: \leq \:
 \begin{pmatrix}
  \hat{S}_{0}(B) & \vec{u} \\ \vec{w}^{t} & r
 \end{pmatrix} \: \leq \:
            S \begin{pmatrix}
               B & \vec{0} \\ \vec{0}^{t} & 0
              \end{pmatrix} \: = \:
 \begin{pmatrix}
  \frac{1}{2} (\text{Tr} B) \, \mathbf{1}_{2} & \vec{0} \\ \vec{0}^{t} & 0
 \end{pmatrix},
\end{equation}
and hence $r=0$, $\vec{u}=\vec{w}=\vec{0}$.
Because any matrix $B \in M_{2}$ is a complex combination of four positive
matrices, we have that indeed $S_{0}(A) \in M_{2} \! \subset \! M_{3}$
for any $A \in M_{2} \! \subset \! M_{3}$.

Let $\{e_{i}\}_{i=1}^{3}$ be the standard orthonormal basis of $\mathbb{C}^{3}$
and $\{ E_{jk} \}_{j,k=1}^{3}$ be the set of matrix units in $M_{3}$,
$E_{jk} = e_{j} e_{k}^{*}$.
We show that for $i = 1,2,3$ and $j=1,2$;
$\langle e_{i}, S_{0}(E_{j3}) e_{i} \rangle = 0$,
and $\langle e_{1}, S_{0}(E_{j3}) e_{2} \rangle = 
    \langle e_{2}, S_{0}(E_{j3}) e_{1} \rangle = 0$.
Fix $i,j$, and take 
$X = |z_{1}|^{2} P_{j} + z_{1} \overline{z_{2}} E_{j3} +
\overline{z_{1}} z_{2} E_{3j} + |z_{2}|^{2} P_{3}$,
for some $z_{1}, z_{2} \in \mathbb{C}$.
It is evident that $X \geq 0$. 
Hence, since $S_{0}(E_{j3}) = S_{0}(E_{3j})^{*}$,
\begin{equation}
0 \: \leq \: \langle e_{i}, S_{0}(X) e_{i} \rangle \: = \:
|z_{1}|^{2} \, \delta_{ij} +
2 \, \text{Re} \, z_{1} \overline{z_{2}} \,
    \langle e_{i}, S_{0}(E_{j3}) e_{i} \rangle +
|z_{2}|^{2} \, \delta_{i3},
\end{equation}
where $\delta_{ij}$ is the Kronecker delta.
Because the above equation is true for every $z_{1}, z_{2} \in \mathbb{C}$,
and $j \neq 3$,
it follows that necessarily 
$\langle e_{i}, S_{0}(E_{j3}) e_{i} \rangle =0$. 
Now, since $X \geq 0$, then so is $Y = P_{12} \, S_{0}(X) P_{12}$.
From what has been just shown, 
$Y_{jj} = |z_{1}|^{2}$,
but the second diagonal element of $Y$ equals 0, and
hence $Y_{12} = Y_{21} = 0$,
which means that
$\langle e_{1}, S_{0}(E_{j3}) e_{2} \rangle
    = \langle e_{2}, S_{0}(E_{j3}) e_{1} \rangle = 0$.

Therefore, we can write that for any
$A  =   \left( \begin{smallmatrix}
    B & \vec{u} \\
    \vec{w}^{t} & z
    \end{smallmatrix} \right) \in M_{3}$,
\begin{equation}
    S_{0}(A) \:=\: S_{0} \begin{pmatrix}
    B & \vec{u} \\
    \vec{w}^{t} & z
    \end{pmatrix} \: = \:
    \begin{pmatrix}
        \hat{S}_{0}(B) & S_{1} \vec{u} + S_{2} \vec{w} \\
        (\overline{S}_{1} \vec{w} + \overline{S}_{2} \vec{u})^{t} & \alpha z
    \end{pmatrix},
\end{equation}
and $S_{1}, S_{2} \in M_{2}$, with matrix elements given by
\begin{equation}
(S_{1})_{ij} = \langle e_{i}, S(E_{j3}) e_{3} \rangle, \quad \quad
(S_{2})_{ij} = \langle e_{i}, S(E_{3j}) e_{3} \rangle.
\end{equation}

Suppose now that $\alpha = 0$.
Then, since $S_{0} P_{\eta} \geq 0$ for every
$\eta \in \mathbb{C}^{3}\backslash\{0\}$,
we have that $S_{1} = S_{2} = 0$.
Because 
$0 \leq (S - S_{0}) P(1,1,1)$,
\begin{equation}
0 \: \leq \: \hat{S}_{0} \begin{pmatrix}
    1 & 1 \\ 1 & 1
    \end{pmatrix} \: \leq \:
    \mathbf{1}_{2} - \frac{1}{2} \begin{pmatrix}
                1 & 1 \\ 1 & 1
                \end{pmatrix} \: = \:
    \frac{1}{2}
    \begin{pmatrix}
    1 & -1 \\ -1 & 1
    \end{pmatrix},
\end{equation}
i.e. for some $\beta \geq 0$, 
$
\hat{S}_{0} \left( \begin{smallmatrix}
    1 & 1 \\ 1 & 1
    \end{smallmatrix} \right) =
    \beta \left(
    \begin{smallmatrix}
    1 & -1 \\ -1 & 1
    \end{smallmatrix} \right)
$.
Repeating the same calculation, but this time for $P(i,i,1)$,
we obtain that for some $\beta' \geq 0$, 
$
\hat{S}_{0} \left( \begin{smallmatrix}
    1 & 1 \\ 1 & 1
    \end{smallmatrix} \right) =
    \beta' \left(
    \begin{smallmatrix}
    1 & 1 \\ 1 & 1
    \end{smallmatrix} \right)
$,
and so 
$
\hat{S}_{0} \left( \begin{smallmatrix}
    1 & 1 \\ 1 & 1
    \end{smallmatrix} \right) = 0
$.
Similarly,
$
\hat{S}_{0} \left( \begin{smallmatrix}
    1 & -1 \\ -1 & 1
    \end{smallmatrix} \right) = 0
$, and hence
$\hat{S}_{0}(\mathbf{1}_{2}) = 0$, i.e. $\hat{S}_{0} = 0$.
Therefore,  we can assume in the following that $\alpha > 0$.

Next, by direct computation, one proves that for any
$\vec{\eta} \in \mathbb{C}^{2}$:
\begin{equation}
 \label{eq:ZeroTrace}
 \text{Tr} \, P(-2\vec{\upsilon}, ||\vec{\eta}||^{2}) \, S P(\vec{\eta},1)
     \:=\: 0,
\end{equation}
where $\vec{\upsilon} =
 \frac{1}{\sqrt{2}} \left(
  \hat{P}_{1} \vec{\eta} + \hat{P}_{2} \overline{\vec{\eta}}
 \right)$,
$||\vec{\upsilon}||^{2} = \tfrac{||\vec{\eta}||^{2}}{2}$.
Because $0 \leq S_{0} \leq S$, the eq.
\eqref{eq:ZeroTrace} holds also for $S_{0}$.
Writing out the equation, and assuming from now on that
$||\vec{\eta}|| = 1$,
we get
\begin{equation}
 \label{eq:ZeroTraceExplicit}
 4 \, \vec{\upsilon}^{\,*} \hat{S}_{0}(\vec{\eta} \vec{\eta}^{\,*}) \vec{\upsilon} -
 4 \, \text{Re} \, \vec{\upsilon}^{\,*} \vec{\upsilon}_{0} + \alpha \: = \: 0,
\end{equation}
where $\vec{\upsilon_{0}} = S_{1} \vec{\eta} + S_{2} \overline{\vec{\eta}}$.
To keep the notation simple,
we remember that $\vec{\upsilon}$ and $\vec{\upsilon}_{0}$ depend on $\vec{\eta}$,
without signifying it explicitly.

Since $S_{0} P(\vec{\eta},1)$ is a positive matrix, 
using again the Schur complement,
$\vec{\upsilon}_{0} \vec{\upsilon}_{0}^{\,*} \leq  \alpha \hat{S}_{0}(\vec{\eta} \vec{\eta}^{\,*})$,
and thus
$|\vec{\upsilon}^{\,*} \vec{\upsilon}_{0}|^{2} \leq
  \alpha \vec{\upsilon}^{\,*} \hat{S}_{0}(\vec{\eta} \vec{\eta}^{\,*}) \vec{\upsilon}$.
Therefore
\begin{equation}
4 \left( \text{Im} \, \vec{\upsilon}^{\,*} \vec{\upsilon}_{0} \right)^{2} +
\left( 2 \, \text{Re} \, \vec{\upsilon}^{\,*} \vec{\upsilon}_{0} - \alpha \right)^{2} 
 \: \leq \:
4 \alpha \, \vec{\upsilon}^{\,*} \hat{S}_{0}(\vec{\eta} \vec{\eta}^{\,*}) \vec{\upsilon}
    - 4 \alpha \, \text{Re} \, \vec{\upsilon}^{\,*} \vec{\upsilon}_{0} + \alpha^{2}  = 0,
\end{equation}
i.e. $\vec{\upsilon}^{\,*} \vec{\upsilon}_{0} = \frac{\alpha}{2}$.
Substituting both $\vec{\upsilon}$ and $\vec{\upsilon}_{0}$,
we obtain:
\begin{equation}
\label{eq:ComplicatedForEtaPlus}
\vec{\eta}^{\,*} \left( \hat{P}_{1} S_{1} + S_{2}^{t} \hat{P}_{2} \right) \vec{\eta}
    + \vec{\eta}^{\,*} \hat{P}_{1} S_{2} \overline{\vec{\eta}} 
    + \vec{\eta}^{\,t} \hat{P}_{2} S_{1} \vec{\eta} \: = \: \frac{\alpha}{\sqrt{2}}.
\end{equation}
Because this holds true for any $\vec{\eta} \in \mathbb{C}^{2}$,
we can repeat the whole argument,
but this time changing $\vec{\eta} \mapsto i \vec{\eta}$,
to obtain:
\begin{equation}
\label{eq:ComplicatedForEtaMinus}
\vec{\eta}^{\,*} \left( \hat{P}_{1} S_{1} + S_{2}^{t} \hat{P}_{2} \right) \vec{\eta}
    - \vec{\eta}^{\,*} \hat{P}_{1} S_{2} \overline{\vec{\eta}}
    - \vec{\eta}^{\,t} \hat{P}_{2} S_{1} \vec{\eta} \: = \: \frac{\alpha}{\sqrt{2}}.
\end{equation}
Adding together \eqref{eq:ComplicatedForEtaPlus} and \eqref{eq:ComplicatedForEtaMinus},
we have
$
\vec{\eta}^{\,*} \left( \hat{P}_{1} S_{1} + S_{2}^{t} \hat{P}_{2} \right) \vec{\eta} = 
    \frac{\alpha}{\sqrt{2}},
$
for any $\vec{\eta}$. 
Hence 
\begin{equation}
\label{eq:S1PlusS2EqualsOne}
\hat{P}_{1} S_{1} + S_{2}^{t} \hat{P}_{2} \: = \:
     \frac{\alpha}{\sqrt{2}} \mathbf{1}_{2}.
\end{equation}
Subtracting \eqref{eq:ComplicatedForEtaMinus} from \eqref{eq:ComplicatedForEtaPlus},
we get
\begin{equation}
    \vec{\eta}^{\,*} \hat{P}_{1} S_{2} \overline{\vec{\eta}} 
    + \vec{\eta}^{\,t} \hat{P}_{2} S_{1} \vec{\eta} \: = \: 0.
\end{equation}
Let  $\vec{\eta}$ be each of the following vectors:
$(1,0)$, $(0,1)$,
$(i,0)$, $(0,i)$,
$(\frac{1}{\sqrt{2}},\frac{1}{\sqrt{2}})$, 
$(\frac{1}{\sqrt{2}},\frac{i}{\sqrt{2}})$,
then we can see that $\hat{P}_{1} S_{2} = \hat{P}_{2} S_{1} = 0$.
Combining this with  \eqref{eq:S1PlusS2EqualsOne}, 
we obtain
\begin{equation}
S_{1} = \begin{pmatrix}
    \frac{\alpha}{\sqrt{2}} & z_{0} \\ 0 & 0 
\end{pmatrix}, \quad
S_{2} = \begin{pmatrix}
     0 & 0 \\ - z_{0} & \frac{\alpha}{\sqrt{2}} 
\end{pmatrix}, \quad
z_{0} \in \mathbb{C}.
\end{equation}

Putting $\vec{\upsilon}^{\,*} \vec{\upsilon}_{0} = \frac{\alpha}{2}$
into \eqref{eq:ZeroTraceExplicit}, we get that
$
\vec{\upsilon}^{\,*} \hat{S}_{0}(\vec{\eta} \vec{\eta}^{\,*}) \vec{\upsilon} = 
\frac{\alpha}{4}
$.
Hence 
\begin{equation}
 \vec{\upsilon}^{\,*} \left(
    \alpha \, \hat{S}_{0}(\vec{\eta} \vec{\eta}^{\,*}) -\vec{\upsilon}_{0} \vec{\upsilon}_{0}^{\,*}
  \right) \vec{\upsilon} \: = \: 0,
\end{equation}
and since 
$\vec{\upsilon}_{0} \vec{\upsilon}_{0}^{\,*} \leq  \alpha \hat{S}_{0}(\vec{\eta} \vec{\eta}^{\,*})$,
it must be that
$ \alpha \, \hat{S}_{0}(\vec{\eta} \vec{\eta}^{\,*}) \vec{\upsilon} =
\left( \vec{\upsilon}_{0}^{\,*} \vec{\upsilon} \right)  \vec{\upsilon}_{0}$,
or simply
$ \hat{S}_{0}(\vec{\eta} \vec{\eta}^{\,*}) \vec{\upsilon} = \frac{1}{2} \vec{\upsilon}_{0}$.
Again, 
we can repeat the whole argument,
changing $\vec{\eta} \mapsto i \vec{\eta}$,
adding the obtained result to and subtracting from the previous one, 
and we arrive at:
\begin{subequations}
    \begin{align}
\label{eq:SZeroIsAlmostOneA}
\eta_{1} \, \hat{S}_{0}(\vec{\eta} \vec{\eta}^{\,*}) \vec{e}_{1} \: = \: 
   \left( \frac{\alpha}{2} \eta_{1} + \frac{\sqrt{2}}{2} z_{0} \eta_{2} \right) \, \vec{e}_{1}, \\
\label{eq:SZeroIsAlmostOneB}
\overline{\eta}_{2} \, \hat{S}_{0}(\vec{\eta} \vec{\eta}^{\,*}) \vec{e}_{2} \: = \: 
   \left(- \frac{\sqrt{2}}{2} z_{0} \overline{\eta}_{1} + \frac{\alpha}{2} \overline{\eta}_{2} \right) 
       \, \vec{e}_{2},
    \end{align}
\end{subequations}
where $\vec{\eta} = (\eta_{1}, \eta_{2})$, and 
$\vec{e}_{1} = (1,0), \vec{e}_{2} = (0,1) \in \mathbb{C}^{2}$.
Taking e.g. \eqref{eq:SZeroIsAlmostOneA}
and multiplying it by $\overline{\eta}_{1} \vec{e}_{1}^{\,*}$,
we have
\begin{equation}
0 \: \leq \:
 |\eta_{1}|^{2} \, \vec{e}_{1}^{\,*} \, \hat{S}_{0}(\vec{\eta} \vec{\eta}^{\,*}) \vec{e}_{1} 
 \: = \: 
\frac{\alpha}{2} |\eta_{1}|^{2} + \frac{\sqrt{2}}{2} z_{0} \overline{\eta}_{1} \eta_{2},
\end{equation}
for every $\vec{\eta} \in \mathbb{C}^{2}$, $||\vec{\eta}|| =1$. 
Therefore $z_{0} = 0$, and 
$S_{1} = \frac{\alpha}{2} \hat{P}_{1}$,
$S_{2} = \frac{\alpha}{2} \hat{P}_{2}$.
As a consequence of 
\eqref{eq:SZeroIsAlmostOneA} and \eqref{eq:SZeroIsAlmostOneB},
$\hat{S}_{0} (\vec{\eta} \vec{\eta}^{\,*}) = \frac{\alpha}{2} \mathbf{1}_{2}$
for every $\vec{\eta}$ such that $\eta_{1} \neq 0$ and $\eta_{2} \neq 0$.
For $0 \!<\!\epsilon \!<\! 1$,
take $\vec{\eta}_{\epsilon} = (\sqrt{1 - \epsilon^{2}}, \epsilon)$.
Then $\hat{S}_{0} (\hat{P}_{1}) = 
\lim_{\epsilon \rightarrow 0} \hat{S}_{0}
    (\vec{\eta}_{\epsilon} \vec{\eta}_{\epsilon}^{\,*}) =
\frac{\alpha}{2} \mathbf{1}_{2}$. 
Similarly,
$\hat{S}_{0} (\hat{P}_{2}) = \frac{\alpha}{2} \mathbf{1}_{2}$,
which results in 
$\hat{S}_{0} (\vec{\eta} \vec{\eta}^{\,*}) = \frac{\alpha}{2} \mathbf{1}_{2}$
for every $\vec{\eta} \in \mathbb{C}^{2}$,
$||\vec{\eta}|| = 1$.
This is sufficient to say that 
$\hat{S}_{0}(B) = \frac{\alpha}{2} (\text{Tr} B) \mathbf{1}_{2}$
for any $B \in M_{2}$.
We have shown that for an arbitrary positive map such that
$0 \leq S_{0} \leq S$, 
$S_{0} = \alpha S$ for $0 \leq \alpha \leq 1$,
which means that $S$ is an extremal positive map of $M_{3}$.

It is of interest to note that $S$ is not 2-positive.
Indeed,
we can prove even more and say that 
the map $S$ does not fulfil the Kadison-Schwarz inequality 
for any matrix $B \in M_{3}$
(see Prop.\,4.1 of \cite{choi1980some}).
Take the matrix
$B = P_{12} + E_{32}$.
Then it is easy to verify that 
$S(B^{*} B) - S(B)^{*} S(B)$
is not positive.
Obviously, the map $S$ is
not 2-\emph{co}positive either:
to see this, take $B = P_{12} + E_{31}$.
That, together with the fact that $S$ is extremal,
makes the map atomic, and ends the proof.
\end{proof}

\begin{remark*}
In order to prove extremality of $S$,
we showed that $S_{0} = \alpha S$,
for every $0 \leq S_{0} \leq S$.
But a only weaker assumption is in fact needed.
Indeed, let us assume that $S_{0} \in \mathcal{P}(M_{3})$ and
$\text{Tr} P_{\xi} \, S_{0}(P_{\eta}) = 0$,
for every $\mathbb{C}^{3} \ni \xi,\eta \neq 0$
such that $\text{Tr} P_{\xi} \, S(P_{\eta}) = 0$.
Then $\text{Tr} P_{12} \, S_{0}(P_{3}) = 0$ and
$\text{Tr} P_{3} \, S_{0}(P_{12}) = 0$.
Because $S_{0}$ is positive, so $S_{0}(P_{3}) = \alpha P_{3}$
and $S_{0}(P_{12}) \in M_{2} \! \subset \! M_{3}$.
Then the rest of the proof goes exactly the same as in 
Lemma \ref{lem:SIsExtremal}, 
beginning with \eqref{eq:SMapsM2intoM2} onward.
This proofs that the map $S$
is not only extremal, but also \emph{exposed}.
\end{remark*}

For the map $S$ specified in \eqref{eq:DefinitionOfS}, we have
$K_{S}= \mathbb{C} P_{12} \oplus \mathbb{C} P_{3}$.
To the best of the authors' knowledge,
all the examples of the extremal maps of $M_{3}$
that are not completely or co-completely positive,
which have been so far specified in the literature,
have $K_{S} = \mathbb{C} \, \mathbf{1}$.
The map $S$ would be the first one
with a two-dimensional commutative stable algebra.

\vspace{0.5cm}

Given a positive map $S: M_{n} \rightarrow M_{n}$,
the entanglement witness associated with $S$ is a matrix $W_{S}$
of the tensor matrix algebra $M_{n^{2}} = M_{n} \! \otimes \! M_{n}$,
defined by
\begin{equation}
\label{def:entanglement-witness}
    W_{S} = \sum \limits_{i,j =1}^{n} E_{ij} \otimes S(E_{ij}),
\end{equation}
where $\{ E_{ij}\}_{i,j=1}^{n}$ are standard matrix units.
The theorem by Choi and Jamiołkowski \cite{choi1975completely,jamiolkowski1974effective}
states that the matrix $W_{S}$ is a positive element of $M_{n^{2}}$,
if and only if $S$ is completely positive.
Therefore, for a positive, non-completely positive map $S$
there is at least one density matrix $\rho \in M_{n^{2}}$
such that $\text{Tr}\, W_{S} \rho < 0$.
This density matrix cannot be separable \cite{werner1989quantum},
and hence we say that the entanglement witness \emph{detects} the
entangled state $\rho$. 

It would be of interest to provide the explicit form of the entanglement
witness associated with the map $S$ of eq. \eqref{eq:DefinitionOfS},
and the family of states on the composite quantum system space
$M_{9} = M_{3} \! \otimes \! M_{3}$, detected by $S$.
It is easy to see, that in this case:
\begin{equation}
\label{WS-C}
 W_S =  \left( \begin{array}{ccc|ccc|ccc}
 \frac{1}{2} &  \cdot& \cdot& \cdot& \cdot& \cdot& \cdot& \cdot& \frac{1}{\sqrt{2}} \\
 \cdot& \frac{1}{2} &\cdot& \cdot& \cdot& \cdot& \cdot& \cdot& \cdot\\
 \cdot& \cdot& \cdot & \cdot& \cdot& \cdot& \cdot& \cdot& \cdot  \\ \hline
 \cdot& \cdot& \cdot& \frac{1}{2} & \cdot& \cdot& \cdot& \cdot&  \cdot \\
 \cdot& \cdot& \cdot& \cdot& \frac{1}{2} & \cdot& \cdot& \cdot&  \cdot \\
 \cdot& \cdot& \cdot& \cdot& \cdot& \cdot& \cdot & \frac{1}{\sqrt{2}}& \cdot  \\ \hline
 \cdot& \cdot& \cdot & \cdot& \cdot& \cdot& \cdot& \cdot& \cdot  \\ 
 \cdot & \cdot& \cdot& \cdot& \cdot& \frac{1}{\sqrt{2}}& \cdot& \cdot& \cdot \\
 \frac{1}{\sqrt{2}}& \cdot& \cdot& \cdot& \cdot & \cdot& \cdot& \cdot& 1
  \end{array} \right),
\end{equation}
where dots mean matrix elements equal to zero.
$W_{S}$ is not a positive matrix of $M_{9} = M_{3} \! \otimes \! M_{3}$.
Let $U \in \text{U}(9)$ be a unitary matrix such that
$U^{*} W_{S} U = W_{S}^{(2)} \oplus W_{S}^{(7)}$,
where
\begin{equation}
\label{WS-DirectSum}
 W_{S}^{(2)} = \left(\begin{array}{cc}
    0 & \frac{1}{\sqrt{2}}  \\ \frac{1}{\sqrt{2}}  & 0
  \end{array}\right) \in M_{2}, \quad
 W_S^{(7)}\ = \  \left( \begin{array}{ccccccc}
 \frac{1}{2} &  \cdot& \cdot& \cdot& \cdot& \cdot& \frac{1}{\sqrt{2}} \\
 \cdot& \frac{1}{2} &\cdot& \cdot& \cdot& \cdot& \cdot\\
 \cdot& \cdot& \cdot & \cdot& \cdot& \cdot& \cdot \\ 
 \cdot& \cdot& \cdot& \frac{1}{2} & \cdot&  \cdot& \cdot \\
 \cdot& \cdot& \cdot& \cdot& \frac{1}{2} & \cdot& \cdot \\
 \cdot& \cdot& \cdot & \cdot& \cdot& \cdot& \cdot  \\ 
 \frac{1}{\sqrt{2}}& \cdot& \cdot& \cdot& \cdot& \cdot& 1
  \end{array} \right) \in M_{7}.
\end{equation}
It is evident that $W_{S}^{(7)} \geq 0$. 
Let also $\vec{v}$ be one of the eigenvectors of $W_{S}^{(2)}$: 
$\vec{v} = \frac{1}{\sqrt{2}} (1,-1)^{t} \in \mathbb{C}^{2}$,
$W_{S}^{(2)} \vec{v} = - \frac{1}{\sqrt{2}} \vec{v}$, and
$P_{\vec{v}} = \vec{v} \vec{v}^{*}$
be the orthogonal projection onto the space spanned by $\vec{v}$.
We define
$\rho = \frac{1}{2} U ( P_{\vec{v}} \oplus \rho_{0} ) U^{*}$,
for a density matrix $\rho_{0} \in M_{7}$, 
$\rho_{0} \geq 0$, $\text{Tr} \rho_{0} = 1$, 
such that
$\text{Tr}\, W_{S}^{(7)} \rho_{0} < \frac{1}{\sqrt{2}}$.
Then $\rho \geq 0$ and $\text{Tr} \rho = 1$,
i.e. $\rho$ is a density matrix.
Moreover,
\begin{multline}
\label{RandomLabel:851252}
    \text{Tr}\, W_{S} \rho =
\frac{1}{2} \text{Tr}\, W_{S} U (P_{\vec{v}} \oplus \rho_{0}) U^{*} =
\frac{1}{2} \text{Tr}\, (W_{S}^{(2)} \oplus W_{S}^{(7)}) (P_{\vec{v}} \oplus \rho_{0}) = \\ =
\frac{1}{2} \text{Tr}\, W_{S}^{(2)} P_{\vec{v}} + \frac{1}{2} \text{Tr}\, W_{S}^{(7)} \rho_{0}
< 0,
\end{multline}
which means that $W_{S}$ detects the entangled state $\rho$.
In particular, the state $\rho$ given by
\begin{equation}
\label{PPTstate}
 \rho =  \frac{1}{7} \left( \begin{array}{ccc|ccc|ccc}
 1 &  \cdot& \cdot& \cdot& \cdot& \cdot& \cdot& \cdot& -1 \\
 \cdot& \cdot& \cdot& \cdot& \cdot& \cdot& \cdot& \cdot& \cdot\\
 \cdot& \cdot& 1 & \cdot& \cdot& \cdot& \cdot& \cdot& \cdot  \\ \hline
 \cdot& \cdot& \cdot& \cdot& \cdot& \cdot& \cdot& \cdot&  \cdot \\
 \cdot& \cdot& \cdot& \cdot& 1 & \cdot& \cdot& \cdot&  \cdot \\
 \cdot& \cdot& \cdot& \cdot& \cdot& 1 & \cdot & -1 & \cdot  \\ \hline
 \cdot& \cdot& \cdot & \cdot& \cdot& \cdot& 1& \cdot& \cdot  \\ 
 \cdot & \cdot& \cdot& \cdot& \cdot& -1 & \cdot& 1 & \cdot \\
 -1 & \cdot& \cdot& \cdot& \cdot & \cdot& \cdot& \cdot& 1
  \end{array} \right)
\end{equation}
is a PPT state, meaning that $(I \otimes t)\rho$ is still a density matrix,
but $\text{Tr}\, W_{S} \rho = \frac{2}{7} - \frac{2\sqrt{2}}{7} < 0$.

\paragraph{Acknowledgements.}
The authors would like to express their gratitude to Prof. Erling Størmer
for rightfully pointing out that the statement of Theorem \ref{thm:FromESbook}
is a direct consequence of Proposition 2.2.10 from
\cite{stormer2013positive}.
Moreover, he helped us complete the list of all possible Jordan subalgebras
of $M_{3}$, which has greatly simplified the proof 
of Theorem \ref{thm:ExposedMaps}.


\bibliographystyle{abbrv}
\bibliography{./biblio}

\end{document}